\documentclass[12pt]{iopart}

\usepackage{float}
\usepackage{graphicx}
\usepackage{url}
\usepackage{amsthm}
\usepackage{caption}
\usepackage{subcaption}

\graphicspath{{./Figures/}}

\newtheorem{theorem}{Theorem}[section]
\newtheorem{corollary}[theorem]{Corollary}

\begin{document}

\title{The Small-World Effect for Interferometer Networks}
\author{Benjamin Krawciw$^{1}$, Lincoln D. Carr$^{1, 2, 3}$, Cecilia Diniz Behn$^{1}$}

\address{$^1$ Department of Applied Mathematics and Statistics, Colorado School of Mines, Golden, CO, USA}
\address{$^2$ Department of Physics, Colorado School of Mines, Golden, CO, USA}
\address{$^3$ Quantum Engineering Program, Colorado School of Mines, Golden, CO, USA}
\ead{benjamin.krawciw@gmail.com}
\vspace{10pt}
\begin{indented}
\item[]October 2023
\end{indented}


\begin{abstract}


Complex network theory has focused on properties of networks with real-valued edge weights. However, in signal transfer networks, such as those representing the transfer of light across an interferometer, complex-valued edge weights are needed to represent the manipulation of the signal in both magnitude and phase. These complex-valued edge weights introduce interference into the signal transfer, but it is unknown how such interference affects network properties such as small-worldness. To address this gap, we have introduced a small-world interferometer network model with complex-valued edge weights and generalized existing network measures to define the \emph{interferometric clustering coefficient}, the \emph{apparent path length}, and the \emph{interferometric small-world coefficient}. Using high-performance computing resources, we generated a large set of small-world interferometers over a wide range of parameters in system size, nearest-neighbor count, and edge-weight phase and computed their interferometric network measures. We found that the interferometric small-world coefficient depends significantly on the amount of phase on complex-valued edge weights: for small edge-weight phases, constructive interference led to a higher interferometric small-world coefficient; while larger edge-weight phases induced destructive interference which led to a lower interferometric small-world coefficient. Thus, for the small-world interferometer model, interferometric measures are necessary to capture the effect of interference on signal transfer. This model is an example of the type of problem that necessitates interferometric measures, and applies to any wave-based network including quantum networks. 
\end{abstract}

\vspace{2pc}

\submitto{{\it J. Phys.: Complexity\/}}

\section{Introduction}
Complex network theory has been used to describe large interacting systems in diverse contexts including sociology \cite{Anato-2007, PadgettJohnF.1993RAat}, the analysis of technological networks like electrical grids \cite{Watts03}, the internet \cite{FaloutsosMichalis1999Opro, 10.1117/12.434393}, and the brain \cite{BullmoreEd2009Cbng, Avena-KoenigsbergerAndrea2018Cdic}. However, complex network theory currently lacks the tools to account for systems with interfering signals. This is especially relevant in problems like quantum networks, where complex-valued edge weights naturally occur. Previous work \cite{SundarBhuvanesh2018Cdot, ZamanShehtab2019RVoQ, BagrovA.A2020Dqcp} handled networks with complex-valued edge weights by taking norms to produce real-valued edge weights, and then applying real-valued complex network measures. This allowed conclusions to be drawn about the magnitude of signals, but this treatment neglected the information stored in the phase of those edge weights. Multiple generalizations of network measures to complex-valued edge weights are possible, and different generalizations may be optimal in different applications. For example, Bottcher and Porter recently introduced alternative generalizations of network measures to complex-valued edge weights \cite{Bot22}, but their measures are not tailored to problems with interfering signals. In particular, the local measures they proposed (strength and clustering) take the form of averages of complex values, which do not involve interference between multiple paths. The discussion on matrix powers and walks does involve interference, but that discussion does not culminate in the introduction of interferometric network measures like the apparent path length measure introduced in this article. A complete treatment of interfering problems requires new network measures that incorporate the phase of complex-valued edge weights on multiple paths as those paths interfere.  


In this article, we take a first step to address this gap in the field of complex networks by extending the concept of small-worldness to a network with complex-valued edge weights that produce interference. We start by modifying the Watts-Strogatz small-world network model \cite{Wat98}, assigning the edges in the network a variable phase $\phi$. The traditional analysis of the small-world model uses two principal network measures: the mean local clustering coefficient and the mean shortest path length between two vertices. The small-world effect occurs when networks simultaneously have short path lengths, on the order of the logarithm of the total network size, while still having a clustering coefficient near one \cite{Wat98, New03}. These two measures can be combined to form a small-world coefficient \cite{Hum08}. These measures, as traditionally defined, do not incorporate phase. Thus, as phase $\phi$ is introduced to edge weights, they will report no change. However, the actual signals at vertices in an interfering small-world network will change with the addition of phase because these signals will undergo constructive or destructive interference. Our extended measures address this discrepancy and are designated as ``interferometric measures" to emphasize their application in the context of interfering signals.  

As a test bed for understanding how signals behave in networks with complex-valued edge weights, we introduce \emph{interferometer networks}. Interferometers are measuring devices that work by splitting waves such as beams of light, allowing those waves to undergo differing phase shifts, and then recombining the waves, causing them to interfere. The intensity of the recombined wave is measured, allowing the user to calculate the difference between the phase shifts associated with distinct paths across the network. We imagined creating a large interferometer with arbitrarily many waves of light, beam splitters, phase shifters, attenuators, and measuring devices (i.e., observers). Such an interferometer is a network over which a light signal is transferred. Based on such an experimental design, which is realizable in the lab (e.g., on an optics table with classical light or in a quantum network experiment), we define the formalism for signal transfer in interferometer networks as a linear algebra problem involving a complex-valued adjacency matrix. The form of the linear algebra of interferometer networks is quite general; interferometer networks serve as an archetype for all network problems involving signal transfer with interference. Thus, interferometer networks can be adapted to other complex-valued signal transfer problems, such as the time evolution of state vectors in quantum walks \cite{Fac13, LuDawei2016Cqw, goldsmith2023link}; inputs, states, and observables in complex-valued observability and controlability problems \cite{LiuYang-Yu2013Oocs, MontanariArthurN2022Foat}; and the matrix analysis of node voltages in alternating-current circuits with complex impedance \cite{Scherz2013, HANCOCK197421}.  We emphasize that at this stage of the analysis of such networks only single-particle or wave-based quantum mechanics is being considered; entangled many-body quantum networks present a future research direction. 


 Next, we generalize the traditional network measures of the clustering coefficient and path length to the \emph{interferometric clustering coefficient} and the \emph{apparent path length}, respectively. Both of these measures incorporate phase by measuring how complex-valued signals add together constructively and destructively in the network context. Using these extended measures, we further define an \emph{interferometric small-world coefficient} to apply to the small-world interferometer model. 
 
 Lastly, we report the results of applying these generalized measures to the small-world interferometer network model in a suite of computational tests. The results demonstrate a rich, phase-dependent behavior in small-worldness that the traditional measures do not capture. 


\section{Small-World Interferometer Model}
To analyze phase-dependence in the small-world effect, we modified the Watts-Strogatz small-world model \cite{Wat98} with complex-valued edge weights. As in the original small-world model, the complex-valued small-world model begins with edges connecting vertices in a ring, and then edges are reshuffled according to a probability $\beta$. Unlike the original small-world model, our model is directed and complex-weighted. First a directed network is constructed by drawing edges out from each vertex, and then the edges are weighted based on an attenuation parameter $s$, out degree $k$, and phase $\phi$. When an edge is reshuffled, the source vertex stays the same, but its destination is randomized. The model is depicted in Figure \ref{fig:ws}, with $N=6$ the number of vertices in the network. The total output strength of vertices is $s$, which must be set such that $s\leq 1$ to control feedback, per Corollary \ref{thm:Pbound}, see Appendix. The out-degree of each vertex is $k$; variable edge weight phase is $\phi$; and $\beta$ is the probability that an edge's destination is randomly reshuffled, in accord with the usual Watts-Strogatz model.

At $\beta = 0$, the model produces a ring, and at $\beta = 1$, the model produces a random network. For $0 < \beta < 1$, the model produces networks that are neither rings nor random, and some networks in this region exhibit the small-world effect; the dominant ring-like structure induces a high clustering coefficient, while the small number of random, long-distance connections greatly reduce the average shortest path length between vertices \cite{Wat98}. 

\begin{figure}[h]
     \centering
     \begin{subfigure}[b]{0.3\textwidth}
         \centering
         \includegraphics[width=\textwidth]{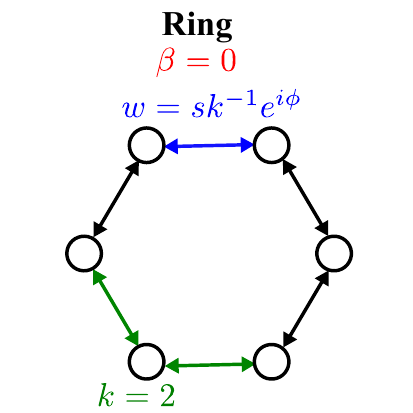}
         \caption{The ring case}
         \label{fig:ringModel}
     \end{subfigure}
     \hfill
     \begin{subfigure}[b]{0.3\textwidth}
         \centering
         \includegraphics[width=\textwidth]{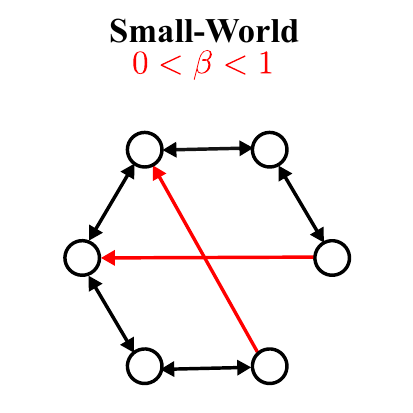}
         \caption{The intermediate region}
         \label{fig:swModel}
     \end{subfigure}
     \hfill
     \begin{subfigure}[b]{0.3\textwidth}
         \centering
         \includegraphics[width=\textwidth]{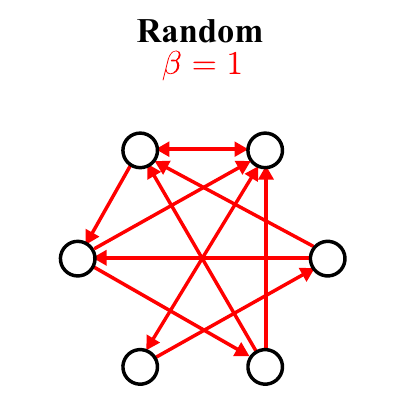}
         \caption{The random case}
         \label{fig:randModel}
     \end{subfigure}
        \caption{Representative small-world interferometer model with $N=6$ vertices, and $k=2$ connections. At $\beta = 0$ (Figure \ref{fig:ringModel}), a ring is formed with $N=6$ vertices connected to their two nearest neighbors. An example of nearest neighbor connections is shown in green to elucidate the meaning of $k$, while the arrowheads indicate the direction of the edges. The edges are weighted with $w = sk^{-1}e^{i\phi}$, highlighted on a particular vertex in blue. For nonzero $\beta$ values (Figure \ref{fig:swModel}), edges are randomly rewired with probability $\beta$. Rewired edges are drawn in red. At $\beta = 1$ (Figure \ref{fig:randModel}), the model yields a kind of random network, where each vertex has out degree 2, but the destinations of those edges are randomized.}
        \label{fig:ws}
\end{figure}

\section{Interferometer Networks}
To analyze the small-world interferometer model, and other problems of its type, we must define this class of problems and the notation for them. We use the case of classical light-based interferometry to inform our decisions. In this case, the signals are the electric field strength at each vertex. We will use this example for context and convenience throughout the rest of the work, but all results are generalizable to arbitrary waves with amplitude and phase, including the Schr\"odinger wavefunction, as found for example in the continuous wave atom laser \cite{chen2022continuous}.

We define interferometer networks to be directed networks with edges weighted by a complex number. The weighted adjacency matrix $W$ contains these complex edge weights. Each vertex has an associated value, corresponding to a signal (the electric field strength). The vertex indexed at $i$ has a signal value $E_i$. The signal vector $\vec{E}$ contains the signals at each vertex, where the vector here refers not to the three spatial components of the electric field but to the number of vertices, $i\in \{1,\ldots,N\}$. The signal $E_i$ is the sum of two inputs: signals traveling over edges to vertex $i$ and a constant source term. The incoming edges carry a signal equal to the edge weight $W_{ij}$ multiplied by the incident vertex’s signal $E_j$. The constant source terms, $S_i$ for each vertex $i$, are contained in a source vector $\vec{S}$. In total, this produces Eq.~(\ref{eqn:vertexSigIndex}).
\begin{equation}
    E_i = S_i + \sum_{j} W_{ij} E_j.
    \label{eqn:vertexSigIndex}
\end{equation}
The entire system is then described by the vertex signal equation, a matrix equation given by Eq.~(\ref{eqn:vertexSig}), :
\begin{equation}
    \vec{E} = W \vec{E} + \vec{S}.
    \label{eqn:vertexSig}
\end{equation}
Eq.~(\ref{eqn:vertexSig}) can be expressed in terms of the network Laplacian, linking this interferometric model to walks. This connection is explored further in Section \ref{sec:conclusion}.

\begin{figure}[ht]
    \centering
    \includegraphics[width = 0.6 \textwidth]{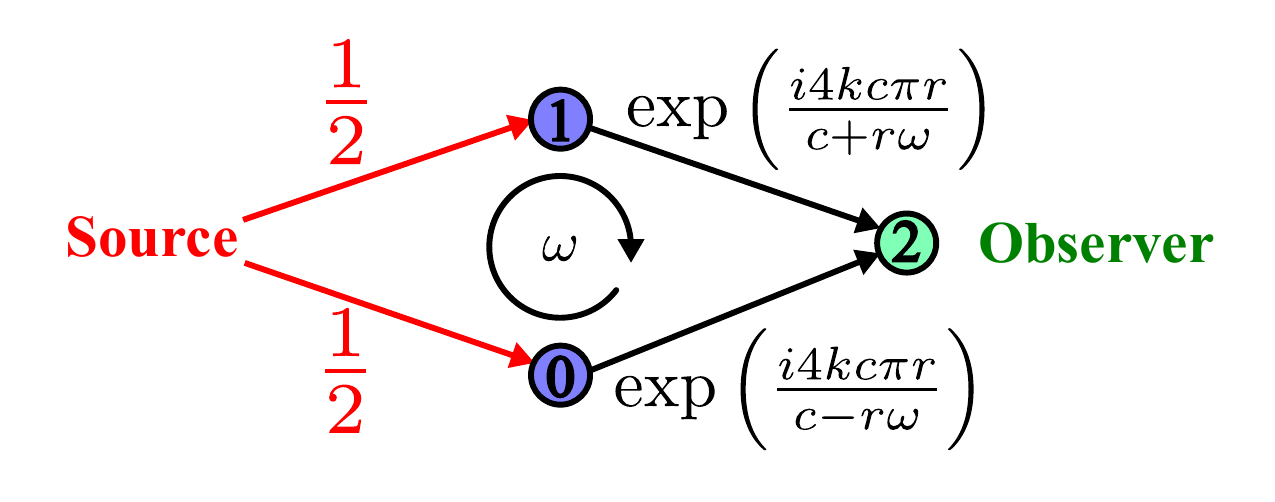}
    \caption{The Sagnac interferometer expressed as an interferometer network. On the network diagram, the source is indicated with red text and lines, the blue vertices are intermediary vertices, and the green vertex is the observer vertex. The parameters in the model are the wavenumber $k$, the speed of light $c$, the radius of the interferometer loop $r$, and the interferometer's angular velocity $\omega$.}
    \label{fig:sagnac}
\end{figure}

As a simple example, we have expressed the Sagnac interferometer, a well-known case used in gyroscopy and many other applications~\cite{culshaw2005optical}, as an interferometer network in Figure \ref{fig:sagnac}. The vertex signal equation for this example is
\begin{equation}
    \vec{E} = \left[ \matrix{
        0 & 0 & 0 \cr
        0 & 0 & 0 \cr
        \exp \left[ \frac{i4kc \pi r}{c+r\omega}\right] & \exp \left[ \frac{i4kc \pi r}{c-r\omega}\right] & 0} 
        \right]
    \vec{E} + \left[ \matrix{
        0.5 \cr
        0.5 \cr
        0
    } \right]
\end{equation}
with solution 
\begin{equation}
    \vec{E} = \left[\matrix{
        0.5 \cr
        0.5 \cr
        0.5 \left( \exp \left[ \frac{i4kc \pi r}{c+r\omega}\right] + \exp \left[ \frac{i4kc \pi r}{c-r\omega}\right] \right)
    }\right].
    \label{eqn:sagnacSolution}
\end{equation}
Taking the magnitude of the final entry of Eq.~(\ref{eqn:sagnacSolution}) yields the expected result for a signal transfer across the Sagnac Interferometer \cite{Pascoli17} and demonstrates the equivalence between the interferometer network formalism and the established analysis of the Sagnac Interferometer. 

\section{Generalized Network Measures}
\label{sec:measures}

The original analysis of the small-world property of the Watts-Strogatz network model \cite{Wat98} is based on the network measures of path length and clustering. The small-world coefficient \cite{Hum08} captures the interplay of these two measures to quantify the small-world effect. To analyze the complex-valued small-world interferometer model, we generalized these measures to describe similar features in complex-weighted networks while capturing the interference behavior of interferometer networks. The extension of the clustering coefficient reduces to the traditional clustering coefficient when the angle $\phi = 0$. However, the generalizations of path length, and, therefore, the small-world coefficient, for interferometer networks do not typically reduce to their real-valued counterparts when $\phi = 0$. These features are reflected in our notation and are presented in more detail in the following subsections.

\subsection{Measuring Interferometer Paths}

When generalizing path length to weighted networks involving real or complex-valued weights, one must decide if and how an edge’s weight contributes to the length of its path. One example of a generalization of path length to weighted networks is Eq.~(\ref{eqn:OpsahlLength}) \cite{OPSAHL2010245}, where $\alpha$ is a parameter that describes how much edge weight contributes to signal transfer or detracts from it.
\begin{equation}
    l = \sum_\textrm{\small path} \left( W_{ij} \right)^\alpha.
    \label{eqn:OpsahlLength}
\end{equation}
However, we argue that paths in interferometers are better characterized by a multiplicative \emph{path strength}, which is the product in Eq.~(\ref{eqn:pathStrength}), because the edge weight in an interferometer network amplifies/attenuates and phase-shifts the signal it carries.  Put simply, multiplied exponentials add in their arguments.
\begin{equation}
    p = \prod_\textrm{\small path}  W_{ij}.
    \label{eqn:pathStrength}
\end{equation}
An additive path length measure can be recovered by taking a logarithm of base $w$, where $w$ is some characteristic edge weight, (e.g., a maximum or mean edge weight magnitude) as shown in Eq.~(\ref{eqn:logp}):
\begin{equation}
    l_p = \log_{w}(p).
    \label{eqn:logp}
\end{equation}

However, the path strength of a single path cannot capture interference, which must involve multiple paths. The total signal sent from vertex $j$ to vertex $i$ is the sum of the signals sent over each path. In practice, for all but the simplest networks, this is computationally challenging to calculate directly. However, the vertex-signal equation (Eq.~(\ref{eqn:vertexSig})) can be algebraically manipulated into Eq.~(\ref{eqn:InvVertexSig}) if the inverse $\left( I - W \right)^{-1}$ exists:
\begin{equation}
    \vec{E} = \left( I - W \right)^{-1} \vec{S}.
    \label{eqn:InvVertexSig}
\end{equation}
The entries $\left[(I - W )^{-1}\right]_{ij}$ quantify the total signal transfer from $j$ to $i$. Thus, we call them the \emph{apparent path strength}, $P_{ij}$. We define the related \emph{apparent path length} to be 
\begin{equation}
a_{ij} = \log_w(P_{ij}).
\label{eqn:lP}
\end{equation}
Apparent path length reduces to traditional path length when only one path exists between $i$ and $j$, along which each edge has weight $w$. Apparent path length will have a similar value to shortest path length when the shortest path dominates signal transfer. However, networks with multiple paths between pairs of vertices will exhibit either constructive or destructive interference. This is why $a_{ij}$ does not typically reduce to the shortest path length, even when no complex phases are involved. Even in the case where all phases are equal to zero, and hence all edge weights are nonnegative real numbers, constructive interference is exhibited. 

We can guarantee that $P = (I - W )^{-1}$ exists by requiring that the $\ell_1$ norm of $W$, $\| W \|_1$, is strictly less than 1. Furthermore, this stipulation bounds the entries of $P$ as shown in Eq.~(\ref{eqn:Pbound}). 
\begin{equation}
    P_{ij} \leq \frac{1}{1-\|W\|_1}.
    \label{eqn:Pbound}
\end{equation}
The proofs for the existence and bounding of $P$ are included in the Appendix. Here we can conceptually explain this bound by noting that the condition that $\| W \|_1 < 1$ is equivalent to requiring that the total signal strength out of any vertex is less than the total signal strength entering the vertex. Thus, $\| W \|_1 < 1$ means that signals decay when passing through a vertex instead of growing or passing undisturbed.  Without an amplifier, this is generally the case in real-world interferometer networks.

To analyze the paths on an interferometer network, we prove that the matrix $P = (I - W )^{-1}$ exists, then we compute the apparent path strength. This measure quantifies both the strength of connections between vertices and the way those paths interact with one another. Since previous network analysis uses path length measures instead of path strength measures, we convert from strength to length using Eq.~(\ref{eqn:logp}).

\subsection{Measuring Interference at One Vertex}

We extend the clustering coefficient to interferometer networks by defining  an \emph{interferometric clustering coefficient} that measures local interference occurring on triangles in a network. Interferometer networks are directed, weighted, and complex-valued, unlike the networks to which the standard clustering measure is typically applied \cite{New03}. Each of these features introduces a challenge to extending clustering. 

Interferometer networks are directed, but the clustering coefficient was originally defined for undirected networks \cite{New03}. For directed networks, several types of triangles can form, and those triangles serve different functions. Fagiolo \cite{Fag07} divides these triangles into four classes: cycle, middleman, in, and out. A clustering measure can be defined with any of these triangle types (or combinations thereof), but middleman triangles lend themselves particularly well to an interferometric interpretation. As depicted in Figure \ref{fig:meshTriangle}, a middleman triangle forms two paths between a pair of vertices $j$ and $k$: one direct, which we call the \emph{shortcut}, and one indirect, passing through vertex $i$, which we call the \emph{through-path}. The interferometric clustering at vertex $i$ compares these two paths.

Clustering was  also originally defined only for unweighted networks. For weighted networks, there are a plethora of generalizations for the clustering coefficient \cite{Sar07}. We have chosen to generalize the interferometric clustering coefficient from the weighted clustering coefficient presented in Zhang \& Horvath \cite{Zhang2005AGF}, which acts on real-valued edge weights $w_{ij}$ and takes the form
\begin{equation}
    C_i = \frac{\sum_{j, k, j \neq k} w_{ki} w_{ij} w_{kj} }{ \sum_{j, k, j \neq k} w_{ki} w_{ij}}.
    \label{eqn:ZhangClust}
\end{equation}
This version lends itself to interpretation as a weighted average of the shortcut edge weight, where weight is given by the path strength of the through-path. This approach is justified in the context of interferometer networks, since interference is most important for signal transfer when it takes place between the strongest paths.

\begin{figure}[h]
    \centering
    \includegraphics[width = 0.6 \textwidth]{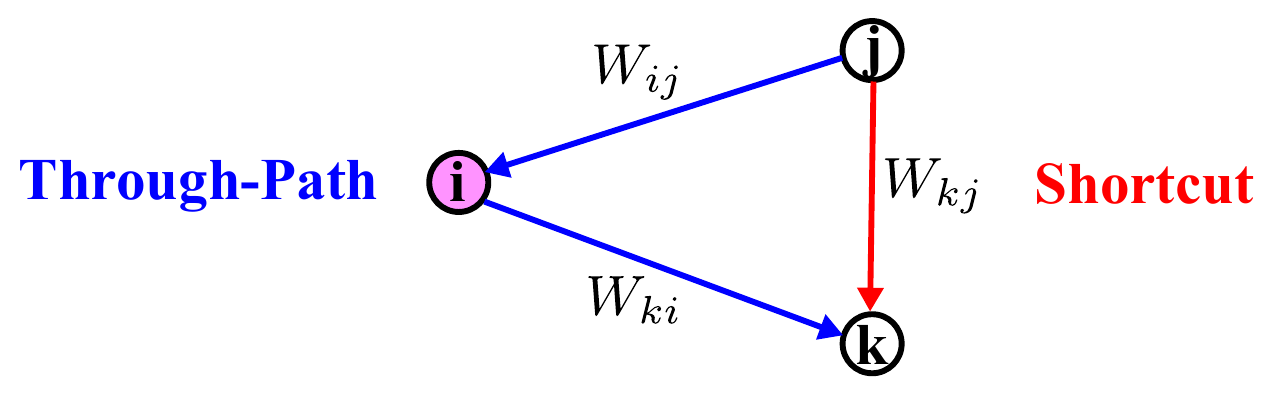}
    \caption{Schematic for the computation of the interferometric clustering at vertex $i$. The path through $i$ is called the \emph{through-path}, and is shown in blue. The path from $j$ to $k$ without $i$ is called the \emph{shortcut}, and is shown in red. Along each edge, the corresponding entry in the complex-valued adjacency matrix is written out in the form $W_{ij}$, to highlight its role in Eq.~(\ref{eqn:intClust}).}
    \label{fig:meshTriangle}
\end{figure}

The definition of the interferometric clustering, $C_i$, is given in Eq.~(\ref{eqn:intClust}): 
\begin{equation}
    C_i = \frac{\sum_{j, k, j \neq k}
    \vert W_{ki} W_{ij} \vert \left( \vert W_{kj} + W_{ki} W_{ij} \vert - \vert W_{ki} W_{ij} \vert\right)
    }{\sum_{j, k, j \neq k}
    \vert W_{ki} W_{ij} \vert
    }
    \label{eqn:intClust}
\end{equation}
The simple $W_{kj}$ in Eq.~(\ref{eqn:ZhangClust}) is replaced by $\left( \vert W_{kj} + W_{ki} W_{ij} \vert - \vert W_{ki} W_{ij} \vert\right)$, which measures how much the magnitude of the the total signal from $j$ to $k$ increases when the shortcut is included. This term is conceptually similar to the reverse triangle inequality in the way that it handles phase differences; if the two paths share the same phase, it reduces to $|W_{kj}|$, but if the two paths have differing phases, the result will be less than $|W_{kj}|$. Note that Eq.~(\ref{eqn:intClust}) is symmetric under exchange of vertices $j,k$. Further, we note that interferometric clustering can take on negative values when the shortcut interferes destructively with the through-path, meaning that the signal from $j$ to $k$ is actually less than if there had been no shortcut at all. The interferometric clustering coefficient reduces to Eq.~(\ref{eqn:ZhangClust}) when the two paths have no phase, and further reduces to the unweighted clustering coefficient when all edge weights equal $1$.

Note that the interferometric clustering coefficient reduces to the weighted clustering coefficient when all $w_{ij}$ entries are nonnegative real numbers. This further reduces to the traditional unweighted clustering coefficient when the network is treated as undirected and all nonzero $w_{ij} = 1$.

\subsection{Measuring the Small-World Coefficient in Interferometer Networks}
A network is considered \emph{small-world} if it has a high clustering coefficients and low vertex-to-vertex path lengths. Humphries \& Gurney \cite{Hum08} defined the small world coefficient, denoted $S_\textrm{real}$ to quantify this property, using random networks of the same size and edge count as a baseline. That measure takes the form
\begin{equation}
    S_\textrm{real} = \frac{\gamma}{\lambda},
    \label{eqn:Strad}
\end{equation}
where 
\begin{equation*}
    \gamma = \frac{\bar{C}}{\bar{C}_{\textrm{random}}}, \lambda = \frac{\bar{l}}{\bar{l}_{\textrm{random}}},
\end{equation*}
$\bar{C}$, $\bar{C}_{\textrm{random}}$ is the mean clustering, and $\bar{l}$, $\bar{l}_{\textrm{random}}$ is the mean shortest path length between two vertices in the network of interest and a random network of the same size, respectively. Thus, large values of $S_\textrm{real}$ correspond to networks with short path lengths, like those in a random network, but also a large clustering coefficient, unlike the random network baseline.

To generalize the small-world coefficient to interferometer networks with complex-valued edge weights, we defined a version of the small-world coefficient that accounts for the possibility of negative interferometric clustering coefficients (arising from destructive interference between shortcuts and through-paths) and negative apparent path lengths (arising from constructive interference that causes net amplification). We define the \emph{interferometric small-world coefficient} to be
\begin{equation}
    S_{\textrm{int}} = \frac{\delta}{\sigma},
    \label{eqn:S}
\end{equation}
where \begin{equation*}
    \delta = \frac{\bar{C} + |\bar{C}_{\textrm{random}}| - \bar{C}_{\textrm{random}}}{|\bar{C}_{\textrm{random}}|},
\end{equation*}
\begin{equation*}
    \sigma = \frac{\bar{a} + |{(\bar{a})}_{\textrm{random}}| - {(\bar{a})}_{\textrm{random}}}{|{(\bar{a})}_{\textrm{random}}|}.
\end{equation*}
The adapted $\delta$ and $\sigma$ definitions were constructed to have the following key properties. For $\delta$,  (1) it reduces to the original definition of $\gamma$ when all inputs are nonnegative numbers; (2) the result is always nonnegative; (3) if $C = C_{\mathrm{random}}$, then $\delta = 1$; (4) if $C > C_{\mathrm{random}}$, then $\delta > 1$; (5) and if $C < C_{\mathrm{random}}$, then $\delta < 1$. Analogous properties hold for $\sigma$. 

The interferometric small-world coefficient $S_{\textrm{int}}$ does not typically reduce to $S_{\textrm{real}}$. This is for the same reason that $a_{ij}$ does not typically reduce to the shortest path length: interference. However, in cases without interference, $S_\textrm{real}$ is recovered from $S_\textrm{int}$. For example, in a real-weighted network where no more than one path exists between any pair of vertices, $a_{ij}$ will reduce to the shortest path length. The inputs to $S_{\textrm{int}}$ will be real and nonnegative, causing $\delta$ to reduce to $\gamma$ and $\sigma$ to reduce to $\lambda$. Thus, these cases recover the traditional small-world coefficient.

\section{Phase Dependence of the Interferometric Small-World Coefficient}
\label{sec:phaseDep}

In this section, we report the results of applying the interferometric clustering, apparent path length, and interferometric small-world coefficient measures to the small-world interferometer model. As a baseline, we also applied the original real-valued measures by taking norms of all edge weights. For real-valued clustering, we use Eq.~(\ref{eqn:ZhangClust}). For real-valued path length, we use the path length recovered from the strongest path strength (Eq.~(\ref{eqn:logp})). For the real-valued small-world coefficient, we used Eq.~(\ref{eqn:Strad}). We describe numerical results for the way $S_\mathrm{int}$ varies with respect to reconnection probability $\beta$ for a few configurations of phase $\phi$, how $\phi$ changes the peak small-world coefficient over all $\beta$, testing resiliency of this effect to non-uniformity in $\phi$, and demonstrating that the observed effect holds over a wide range of number of vertices $N$ and  out-degree $k$. Overall, we find that the interferometric small-world coefficient depends significantly on the phase of edge weights.

First, we examine networks with $N = 500$, $k = 12$, $s = 0.9$, and uniform phase $\phi$ on all edges in the network. We chose $N = 500$ because it was the largest network size our computing cluster could test in large batches in a few hours. Nearest-neighbor count $k = 12$ was chosen to ensure that $k << N$, where the effect of rewiring is most visible \cite{Wat98}, but also so that $k$ was large enough to give a clustering coefficient near $1$. Attenuation parameter $s = 0.9$ was chosen because it is close enough to $1$ that constructive interference on paths can cause strong long-range signal transfer, while it is small enough that the apparent path length does not diverge to extremely large values. For each selected configuration of $\beta$ and $\phi$, we ran at least 100 tests (more for sensitive values of $\beta$ at $\phi = 0$), computed their complex network measures, and averaged them for each set of model parameters. We plot $S_\textrm{int}$ over $\beta$ for a few values of $\phi$ ranging from $0$ to $\pi$ in Figure \ref{fig:overBeta}.

\begin{figure}[H]
    \centering
    \includegraphics[width = 0.95 \textwidth]{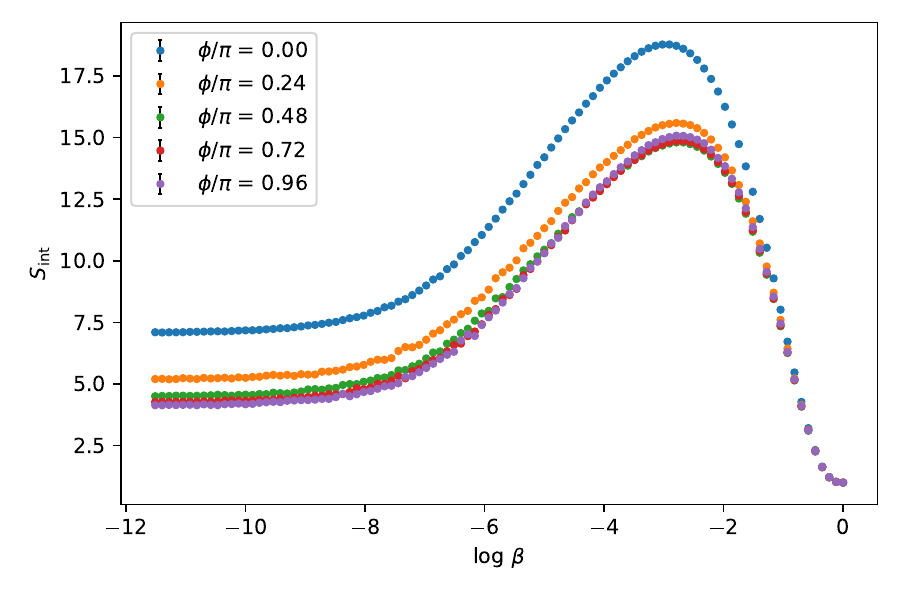}
    \caption{\emph{Uniform-phase small-world interferometer.}  The interferometric small-world coefficient (Eq.~(\ref{eqn:S})) is relatively small at the extremes of $\beta$, but peaks in an intermediate region, where the interferometric clustering coefficient is high while the apparent path length is low, recreating the original small-world effect. However, we see that the scale of this curve changes as $\phi$ is varied: $\phi$ near $0$ or $2\pi$ is dominated by constructive interference, which makes the peak higher, while $\phi$ near $\pi$ introduces more destructive interference, which diminishes the height of the peak of the $S_\textrm{int}$ curve. Here there are $N = 500$ vertices and $k = 12$ nearest-neighbor connections, and the error bars represent the spread due to $\geq 100$ random instances of small-world model networks. Error bars are included, but they are not visible because they are smaller than the circular point markers. }
    \label{fig:overBeta}
\end{figure}

The first key observation from Figure \ref{fig:overBeta} is that the interferometric small-world coefficient quantifies the original small-world effect. At very small values of $\beta$, path lengths are long, but the clustering coefficient is high. At $\beta$ near 1, the clustering coefficient is lower, but the path lengths are short. At both of these extremes of $\beta$, the interferometric small-world coefficient is relatively small. However, there is an intermediate region of $\beta$ where, simultaneously, clustering is high and path lengths are short. This is where the small-world coefficient peaks.  This replicates the behavior of $S_\textrm{real}$ in \cite{Hum08}. Thus, while $S_\textrm{int}$ modifies the form of $S_\textrm{real}$ heavily, it still functions as a small-world coefficient, since its definition reflects the balance between clustering and path length that is captured by the original small-world coefficient.

The second key observation from Figure \ref{fig:overBeta} is that although similar behavior of $S_\textrm{int}$ with respect to $\log \beta$ is observed for each value of $\phi$, $S_\textrm{int}$  also changes with $\phi$. The $S_\mathrm{int}$ vs. $\mathrm{log}\beta$ curve attains its maximum at $\phi=0$, when constructive interference simultaneously strengthens the interferometric clustering coefficient and shortens the apparent path length. As $\phi$ increases, destructive interference is introduced, and the small-world effect is weakened. The overall scale of the curve is reduced as $\phi$ increases. 

To get a clearer picture of this new effect, we measured the peak of the curve in Figure \ref{fig:overBeta} for each value of $\phi$ and plotted it in Figure \ref{fig:overphi}. For reference, we plotted the curve of peak values against the peak value of the real-valued small-world coefficient, which does not change with respect to $\phi$, aside from minor changes due to numerical variation or error. We see that for the $N = 500$, $k = 12$, $s = 0.9$ case, the small-world effect is strengthened by $22 \%$ at $\phi = 0$. However, as $\phi$ increases, the small-world effect is weakened by destructive interference by as much as $5 \%$. This pattern repeats in reverse as $\phi$ approaches $2\pi$, due to the $2\pi$-periodicity of phase. Intuition may suggest that $S_\textrm{int}$ ought to approach $S_\textrm{real}$ as $\phi \to 0$, since the computation of $S_\textrm{int}$ becomes one involving only real numbers. However, this is not the case, because, in addition to involving complex-valued terms, the interferometric small-world coefficient accounts for signal transfer along all possible paths between pairs of vertices. This contrasts with $S_\textrm{real}$, which only accounts for signal transfer along the shortest path. Thus, in general, $S_\textrm{int}$ does not reduce to $S_\textrm{real}$ when $\phi = 0$ because the interferometric measure captures the combined effect of all paths constructively interfering together. Similar phase-dependence was observed in maximum interferometric small-world coefficient values for all network configurations tested (see details below).

\begin{figure}[H]
    \centering
    \includegraphics[width = 0.95 \textwidth]{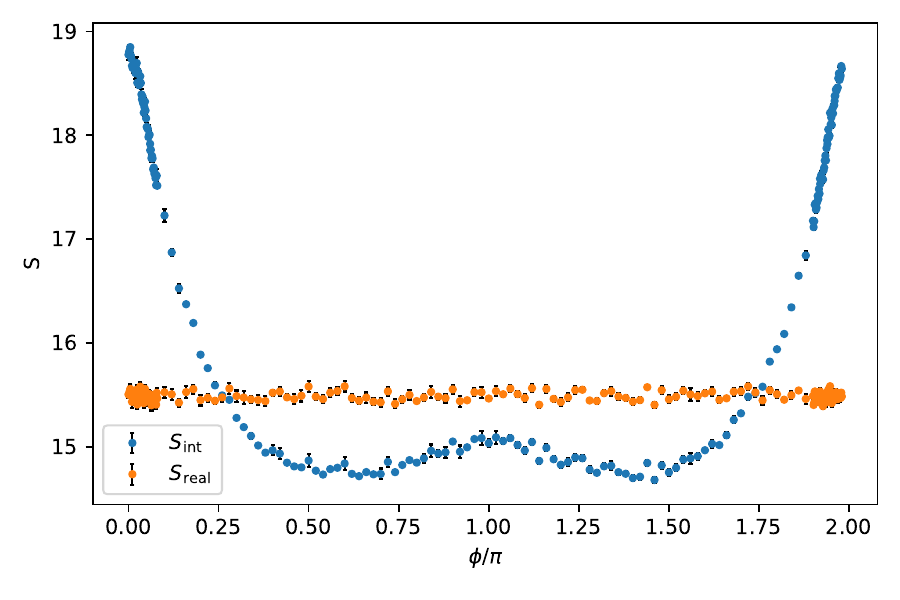}
    \caption{Maximum interferometric small-world coefficient values for each configuration in phase $\phi$. Near $\phi = 0$, the small-world effect is strengthened by constructive interference. Further from $\phi = 0$, destructive interference is introduced, and the small-world effect is weakened. The small-world coefficient increases again at $\phi = 2\pi$, since phase is $2\pi$-periodic. The original small world coefficient $S_\textrm{real}$ (per Eq.~(\ref{eqn:Strad})) is plotted for reference; it is approximately constant because it does not depend on phase. This plot shows results for the small-world interferometer model with size $N = 500$ and $k = 12$ nearest neighbor connections.}
    \label{fig:overphi}
\end{figure}

We tested a version of the $N = 500$, $k = 12$, $s = 0.9$ case with $\phi$ non-uniform, to see if the phase-dependence of the interferometric small-world coefficient was sensitive to small variation in $\phi$. We suspected that the effect might not be resilient to $\phi$ variability, especially near $\phi = 0$, because variability would introduce destructive interference to the case otherwise dominated by strict constructive interference over long paths. This is relevant because an experimental interferometer network will likely have such variability in practice. For this test we added a normally-distributed $\pm 0.2 \pi$ error to all phases in the network. The results in Figure \ref{fig:phaseVar} show that phase variability diminishes the phase-dependence of the interferometric small-world coefficient, but the phase-dependence of the effect remains significant.

\begin{figure}[H]
    \centering
    \includegraphics[width = 0.95 \textwidth]{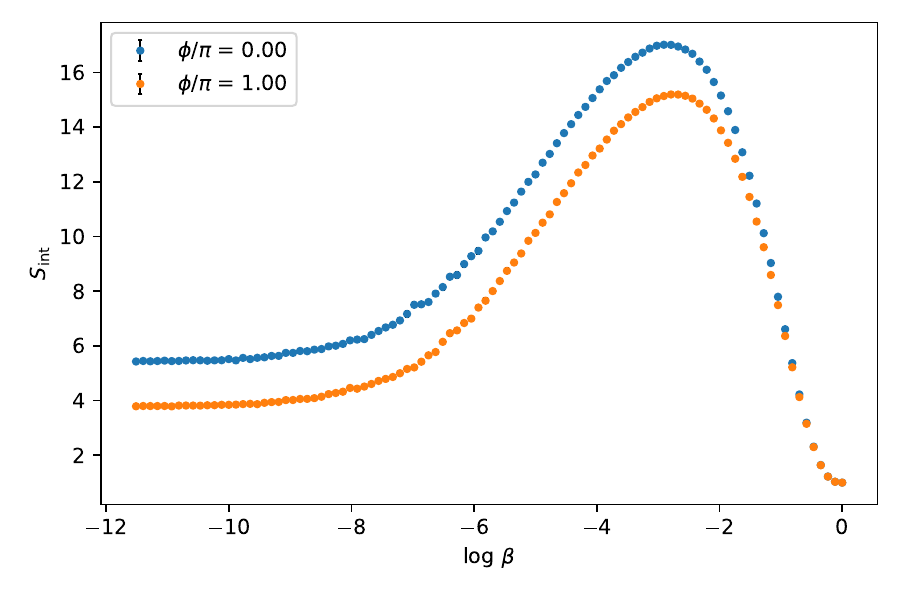}
    \caption{\emph{Small-world interferometer with phase variability.}  The interferometric small-world coefficient $S_\mathrm{int}$ is plotted against $\log(\beta)$ for small-world interferometer networks with $N=500$ nodes, $k=12$ nearest-neighbor connections from each node, and phases on all edges distributed normally with a mean value of $\phi$ and a standard deviation of $10 \%$ of $2 \pi$. We compare mean phases of $\phi = 0$ and $\phi = \pi$. We find that the effect of $\phi$ on the scale of $S_\mathrm{int}$(see Figure \ref{fig:overBeta}) is diminished, but it remains significant.}
    \label{fig:phaseVar}
\end{figure}

While $N = 500$, $k = 12$ serves to demonstrate that the small-world effect can change as $\phi$ varies, it is only a particular case. To demonstrate that this effect holds more generally, we ran tests on a wide range of parameters. For each set of parameters, we ran 100 tests. The parameters were selected with ranges $100 \leq N \leq 1500$ and $4 \leq k \leq 10$. For each configuration of $N$ and $\phi$, we ran 50 trials and averaged their measures. In particular, we examined the interferometric small-world coefficient at $\phi = 0$ and $\phi = \pi$. We selected these values of $\phi$ because $\phi=0$ is the case for which total constructive interference dominates, while $\phi = \pi$ is the center of the region where destructive interference exists. Then, we computed the ratio of these two measurements. Figure \ref{fig:randHist} depicts a histogram of the base-10 logarithm of these ratios. Notice that, for all configurations, the logarithm is greater than zero, which implies that $S_\mathrm{int}(\phi = 0) > S_\mathrm{int}(\phi = \pi)$ for all trials. This means that, like in Figure \ref{fig:overphi}, the small-world coefficient is higher when all interference is constructive than when destructive interference exists at $\phi = \pi$. This holds for all tested configurations of $N$ and $k$ at $s = 0.9$. Examining the modal value in Figure \ref{fig:randHist}, $10^{0.1}\simeq 1.25$ indicates that constructive phase interference typically increases the small-world effect by about 25\% over the destructive case.

\begin{figure}[H]
    \centering
    \includegraphics[width = 0.7 \textwidth]{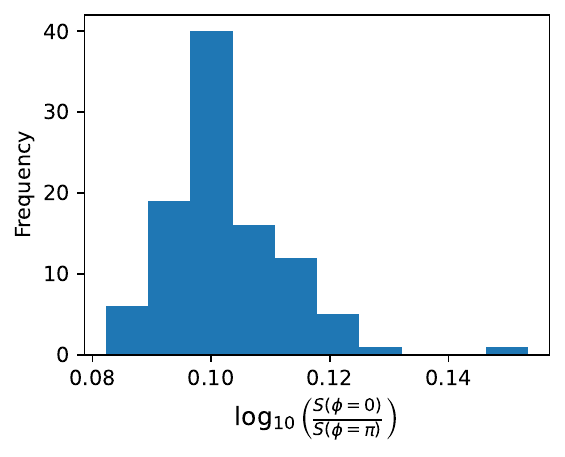}
    \caption{This histogram confirms that the interferometric small-world coefficient, $S_\textrm{int}$, is higher at edge-weight phase $\phi = 0$ than it is at $\phi = \pi$ for all tested configurations of network size $N$ and nearest-neighbor-connection number $k$ with uniform phase. This is shown by binning and counting the logarithms of the ratio $S_\mathrm{int}(\phi = 0) / S_\mathrm{int}(\phi = \pi)$. Notice that the logarithms are always greater than zero, implying  $S_\mathrm{int}(\phi = 0) > S_\mathrm{int}(\phi = \pi)$ for all tests.}
    \label{fig:randHist}
\end{figure}

In all of our tests, constructive interference at phases near $\phi = 0$ strengthens the small-world effect, while destructive interference at phases further from $\phi = 0$ weakens the small-world effect. This behavior is captured by the interferometric measures -- apparent path length, interferometric clustering, and interferometric small-world coefficient -- but it is not detected by the real-valued measures taken by eliminating phase information with an absolute value. 

\section{Discussion and Conclusion}
\label{sec:conclusion}
Our computational tests revealed that the small-world effect is made stronger by constructive interference and weaker by destructive interference. We measured this behavior by applying our newly defined interferometric measures: apparent path length (Eq.~(\ref{eqn:lP})), interferometric clustering (Eq.~(\ref{eqn:intClust})), and the interferometric small-world coefficient (Eq.~(\ref{eqn:S})) to a small-world interferometer model (Figure \ref{fig:ws}). In contrast, the original real-valued measures of path length, clustering, and the small-world coefficient, found by taking absolute values of all complex-valued edge weights, are insensitive to the effects of interference.

This result serves as an example of the type of problem that requires measures, such as interferometric measures, that account for interfering signals. Such problems are ubiquitous in physical science; they include quantum walks \cite{Fac13, LuDawei2016Cqw, goldsmith2023link}, complex-valued observability and controlability problems \cite{LiuYang-Yu2013Oocs, MontanariArthurN2022Foat}, and the matrix analysis of node voltages in alternating-current circuits with complex impedance \cite{Scherz2013, HANCOCK197421}. The interferometer network scheme can be adapted to these contexts by modifying the vertex signal equation (Eq.~(\ref{eqn:vertexSig})) to relate the relevant signals at vertices to whatever quantity is considered the edge weight. For example, time-dependent AC networks based on fixed carrier frequency can effectively be represented by an edge with amplitude (signal strength) and phase (signal phase), providing a new method to study certain simpler classes of network synchronization problems. As written, Eq.~(\ref{eqn:vertexSig}) already bears a resemblance to a steady-state solution to a random walk. We can rewrite Eq.~(\ref{eqn:vertexSig}) as
\begin{equation}
    {\vec{S}} - (I - W)\vec{E} = \vec{0},
\end{equation}
which is the steady-state solution to
\begin{equation}
    \frac{d}{dt} \vec{E} = {\vec{S}} - (I - W)\vec{E}.
\end{equation}
Now, $\vec{S}$ becomes a constant source term, $(I-W)$ becomes a normalized, weighted Laplacian matrix, and the time evolution becomes reminiscent of the walk problems in \cite{Fac13, goldsmith2023link}. Once a physical problem is expressed as an interferometer network, measures such as interferometric clustering and apparent path strength, will be necessary for accounting for the effect of interference on signal transfer.

Beyond recasting other network problems as interferometer networks, this work presents several other opportunities for further research. The most immediate problem is to analytically describe the behavior of the small-world interferometer model with respect to the model parameters. This would give a much more thorough understanding of both the effect described in this article and any others that arise due to the inclusion of phase. The next direction of further study is quantum mechanics. Although this work was performed in the context of interferometry, this was intended to be a first step in applying complex-valued network measures to quantum problems. In particular, interferometric measures lend themselves to quantum walks and condensed matter models ~\cite{BagrovA.A2020Dqcp}. Toward this goal, interferometric measures could be applied to extend recent progress on analyzing walks on complex-weighted networks, especially Hermitian adjacency matrices \cite{Yu23}. A related topic of interest is applying interferometric measures to neural networks for quantum systems undergoing a phase transition; it is possible that the interferometric measures can detect the phase transition. Lastly, the analysis of small-world interferometer networks ought to be modeled with further real-world considerations, especially different edge weight distributions, and the interferometric measures ought to be applied to real-world data sets. 

\section{Acknowledgements}
This material is based upon work supported by the National Science Foundation under Grants CCF-1839232, DMR-2002980, and PHY-221056.  This work was performed in part at Aspen Center for Physics, which is supported by National Science Foundation grant PHY-2210452.

The authors acknowledge useful conversations with Brandon Barton.  We acknowledge the Colorado School of Mines High-Performance Computing resources (\url{https://ciarc.mines.edu/hpc/}) made available for conducting the research reported in this article. 

\appendix
\section*{Appendix: Existence and Bounding of Apparent Path Strength}
\setcounter{section}{1}
\label{app:P}

In Section \ref{sec:measures}, we noted that the apparent path strength $P$ is only defined if the matrix $(I-W)^{-1}$ exists. We posited that requiring $\|W\|_1 < 1$ would be sufficient to guarantee existence, and we gave a conceptual explanation for this requirement. Here, we present the proofs for the existence and bounding of $P$. 

\begin{theorem}[Existence of $P$]
Consider an interferometer network with weighted adjacency matrix $W$ such that $\|W\|_1 < 1$. Then the matrix $I-W$ is invertible, and the apparent path strength matrix 
$P = \left(I - W\right)^{-1}$ exists. 
\label{thm:Pex}
\end{theorem}
\begin{proof}
If it exists, $P$ is the inverse of $(I-W)$. By the fundamental theorem of invertible matrices \cite[172]{PooleDavid2014LAAM}, it will suffice to show that, for all $\vec{x} \neq \vec{0}$,
\begin{equation}
    (I - W)\vec{x} \neq \vec{0}.
\end{equation}
Equivalently, this will be true if
\begin{equation}
    \|(I - W)\vec{x} \|_1 > 0,
\end{equation}
for all $\vec{x} \neq \vec{0}$. By the triangle inequality,
\begin{equation}
    \|(I - W)\vec{x} \|_1 + \|W\vec{x}\|_1 \geq \|\vec{x}\|_1.
\end{equation}
\begin{equation}
    \Rightarrow \|(I-W)\vec{x}\|_1 \geq \|\vec{x}\|_1 - \|W\vec{x}\|_1.
\end{equation}
By the consistency of the $\ell_1$ matrix norm \cite{BeilinaLarisa2017NLAT},
\begin{equation}
    \|W \vec{x}\|_1 \leq \|W\|_1 \|\vec{x}\|_1.
\end{equation}
Introducing the matrix norm into our inequality, we have
\begin{equation}
    \|(I-W)\vec{x}\|_1 \geq \left(1 - \|W\|_1 \right) \|\vec{x}\|_1.
\end{equation}
Therefore, $(I-W)$ is invertible and $P$ exists whenever $1 - \|W\|_1 > 0$.
\end{proof}

\begin{corollary}[Bounding the entries of $P$]
 Furthermore, the entries of the $P$ matrix in Theorem \ref{thm:Pex} are bounded. In particular, let $P_{\max} \equiv \max_{i, j}|P_{ij}|$. Then,
 $$
 P_{\max} \leq \frac{1}{1 - \|W\|_1}.
 $$
 \label{thm:Pbound}
\end{corollary}
\begin{proof}
The norm $\|W\|_1$ is calculated as
\begin{equation}
    \|W\|_1 = \max_{1 \leq j \leq N} \sum_{i = 1}^N |W_{ij}|.
    \label{eqn:matNorm}
\end{equation}
By examining Eq.~(\ref{eqn:matNorm}), we observe that the $\ell_1$ norm of $P$ is an upper bound for $P_{\max}$. The $\ell_1$ norm is defined \cite{BeilinaLarisa2017NLAT} as
\begin{equation}
    \|P\|_1 = \sup_{\vec{y} \neq 0} \frac{\| P \vec{y} \|_1}{\|\vec{y}\|_1}. 
\end{equation}
Let $\vec{x} = P \vec{y}$. Then$\|P\|_1$ can be equivalently expressed as
\begin{equation}
    \|P\|_1 = \sup_{\vec{x} \neq 0} \frac{\|\vec{x}\|_1}{\|(I - W) \vec{x} \|_1}.
\end{equation}
Now, as before, we use the triangle inequality and the consistency of the $\ell_1$ matrix norm \cite{BeilinaLarisa2017NLAT} to show that
\begin{equation}
    \|(I-W)\vec{x}\|_1 \geq \left(1 - \|W\|_1 \right) \|\vec{x}\|_1.
\end{equation}
Therefore,
\begin{equation}
    P_{\max} \leq \|P\|_1 \leq \frac{1}{1 - \|W\|_1}.
\end{equation}
\end{proof}

\section*{References}
\bibliographystyle{unsrt}
\bibliography{myrefs}

\end{document}